\documentclass[smallextended]{svjour3}       
\smartqed  
\usepackage{graphicx,amssymb,amstext,amsmath}
\usepackage[justification=centering]{caption}

\usepackage{epsfig}
\usepackage{subfigure}
\usepackage{amsmath}
\usepackage{algorithmic}
\usepackage{graphics}
\usepackage{amssymb}
\usepackage{alltt}
\usepackage{float}
\usepackage{array}
\usepackage{xspace}
\usepackage{footmisc}
\usepackage{comment}
\usepackage{algorithm}
\usepackage{algorithmic}
\usepackage[misc]{ifsym}

\newcounter{RomanNumber}
\newcommand{\MyRoman}[1]{\setcounter{RomanNumber}{#1}\Roman{RomanNumber}}
\newcommand{\ifa}{IFA\xspace}
\newcommand{\stvii}{STV2I\xspace}
\newcommand{\tdr}{3D-Rtree\xspace}

\newcommand{\sii}{S2I\xspace}
\newcommand{\hiq}{HI-Quadtree\xspace}
\newcommand{\qdtree}{Quadtree\xspace}
\newcommand{\iqdtree}{Inverted Quadtree\xspace}
\newcommand{\pyr}{Pyramid\xspace}
\newcommand{\rtree}{R-tree\xspace}

\begin{document}

\title{Hierarchical Information Quadtree: Efficient Spatial Temporal Image Search for Multimedia Stream
}
\subtitle{}


\author{Chengyuan Zhang         \and
        Ruipeng Chen         \and
        Lei Zhu   \and
        Anfeng Liu      \and
        Yunwu Lin       \and
        Fang Huang     \and 
}


\institute{Chengyuan Zhang \at
              \email{cyzhang@csu.edu.cn}
            \and
            Ruipeng Chen \at
              \email{rpchen@csu.edu.cn}
            \and
            Lei Zhu \at
              \email{leizhu@csu.edu.cn}
            \and
            \Letter Anfeng Liu \at
              \email{anfengliu@csu.edu.cn}
            \and
           Yunwu Lin \at
              \email{lywcsu@csu.edu.cn}
              \and
           Fang Huang \at
              \email{hfang@csu.edu.cn}
           \\
           \\
           School of Information Science and Engineering, Central South University, PR China \\
           Big Data and Knowledge Engineering Institute, Central South University, PR China
}

\date{Received: date / Accepted: date}

\maketitle

\begin{abstract}
Massive amount of multimedia data that contain times-
tamps and geographical information are
being generated at an unprecedented scale in many emerging applications such as photo
sharing web site and social networks applications. Due to their importance, a large body of work has focused on efficiently
computing various spatial image queries. In this paper,we study the spatial temporal image
query which considers three important constraints during the search including time recency, spatial proximity and visual relevance. A novel index structure, namely Hierarchical Information Quadtree(\hiq), to efficiently insert/delete
spatial temporal images with high arrive rates. Base on \hiq an efficient algorithm is developed
to support spatial temporal image query. We show via extensive experimentation with real
spatial databases clearly demonstrate the efficiency of our methods.

\keywords{Hierarchical information quadtree, spatial temporal image search, multimedia stream}
\end{abstract}

\section{Introduction}
\label{intro}

Due to the rapid development of modern web services and popularization of mobile smart devices, massive amount of multimedia data that contain both text information and geographical location information are being generated at an unprecedented scale. For instance, Facebook, the most popular social networks service, reports 350 million photos uploaded daily as of November 2013. More than 400 million tweets containing texts and images have been generated by 140 million Twitter users everyday. Tweets, each containing up to 140 characters, can be associated with locations, which may be coordinates (latitude and longitude) or semantic locations. Flickr, the largest photo sharing web site, had a total of 87 million registered members and more than 3.5 million new images uploaded daily in March 2013. 100 hours of video are uploaded to YouTube every minute, resulting in more than 2 billion videos totally by the end of 2013. Other social networks applications such as WeChat, Instgram, Weibo, Pinterest, generate vast amount of multimedia data day by day and are shared over the world. Mobile smart devices, such as smartphone, tablet and smart watch which are equiped with GPS module and wireless communication module, can take photos, make videos or post messages to social platforms. These multimedia data generally contain timestamps and geographical information. For another example, check-ins or
reviews in location based social networks (e.g., Foursquare) contain both text descriptions and locations of points of interest (POIs). The emergence of massive multimedia lead to the new requirement such as temporal spatial multimedia data searching.

Top-$k$ temporal spatial image queries are intuitive and constitute a useful tool for many applications. It aims to find temporal spatial image objects that attend the three criteria simultaneously: they are similar or relevant in the aspect of visual content, they are inside the spatial area of interest and they are inside the query time constraints. Fig  illustrates the spatial temporal images in a three dimension space: longitude $x$, latitude $y$ and time $t$. However, processing top-$k$ temporal spatial image queries efficiently is complex and requires a hybrid index combining information retrieval and spatial indexes. Besides, It is costly and involves accessing a huge amount of temporal geo-tagged multimedia data before finding the result set. To the best of our knowledge, we are the first to study this important issue and the previous works Focused on top-$k$ spatial image search without temporal information. The state-of-the-art approaches proposed by Alfarrarjeh et al.~\cite{DBLP:journals/corr/AlfarrarjehS17} employ a hybrid index that evaluate both spatial and visual features in tandem.

\textbf{Challenges.} There are two key challenges in efficiently
processing spatial temporal image queries over spatial temporal multimedia streams. Firstly, a massive
number of geo-temporal images, typically in the order of millions, are
posted in many applications, and hence even a small increase
in efficiency results in significant savings. Secondly,
the streaming geo-temporal images may
continuously arrive in a rapid rate which also calls for high
throughput performance for better user satisfaction.

Based on the above challenges, we propose a novel index technique,
Hierarchical Information Quadtree(\hiq for short), to
effectively organize continuous spatial temporal multimedia streams. In a nutshell, \hiq consists of two parts temporal segment and inverted quadtree.
temporal segment makes sure that the newly incoming spatial temporal images are always inserted in the most recent segment, and provides timely answers to spatial temporal image queries. The inverted quadtree is essentially an quadtree, each node of which is enriched
with reference to an inverted file for the images contained in the sub-tree rooted at the
node. Through inverted quadtree, we takes spatial and
visual dimensions into consideration synchronously during query processing. Extensive experiments show that our \hiq based
spatial temporal image search algorithm achieves very substantial improvements
over the nature extensions of existing techniques due to
strong filtering power.


\textbf{Contributions.} The principle contributions of this paper are summarized as follows.
\begin{itemize}
\item We formulate the problem of spatial temporal image query and identify its applications.
\item To facilitate the spatial temporal image search,
we propose a novel indexing structure namely \hiq to effectively organize spatial temporal images.
\item Based on \hiq, we develop an efficient spatial temporal image search algorithm.
\item Comprehensive experiments show that our new matching
algorithm achieves substantial improvements (up to
three to five times speed up) over the nature extensions
of state-of-the-art techniques.
\end{itemize}

\noindent\textbf{Roadmap.}
The rest of this paper is organized as follows. Section~\ref{relwork} introduces the related work. Section~\ref{overview} first gives an overview of architecture of spatial temporal image search system, then formally defines the problem of spatial temporal image search. Baseline technique are presented in Section~\ref{baseline}. We introduce the techniques should be adopted in Section~\ref{StvIndex}.  Extensive experiments are reported in Section~\ref{perform}. Section ~\ref{con} concludes the paper.

\subsection{Related Work}
\label{relwork}

In this section, we review the techniques of top-$k$ spatial keywords query, temporal spatial keywords query and content-based image retrieval, which are related to our work.

\noindent\textbf{Top-$k$ Spatial Keywords Query.} Spatial keywords query is a hot issue attracting a lot of researchers in the community of database and information retrieval. A spatial keyword query takes a user location and user-supplied keywords as arguments and returns web objects that are spatially and textually relevant to these arguments~\cite{DBLP:conf/er/CaoCCJQSWY12}. Many efficient
indexing techniques have been proposed such as R-tree~\cite{DBLP:conf/sigmod/Guttman84}, R$^*$-tree~\cite{DBLP:conf/sigmod/BeckmannKSS90}, IR~\cite{DBLP:journals/pvldb/CongJW09}. For top-$k$ spatial keyword query problem, Jo˜ao B. Rocha-Junior~\cite{DBLP:conf/ssd/RochaGJN11} et al. propose a novel index called Spatial Inverted Index (S2I) which maps each distinct term to a set of objects containing the term. Based on S2I, they designed efficient algorithms named SKA and MKA to improve the performance of top-$k$ spatial keyword queries. Li et al.~\cite{DBLP:journals/tkde/LiLZLLW11} proposed IR-tree that together with a top-$k$ document search algorithm facilitates four major tasks in document searches. Zhang et al.~\cite{DBLP:conf/edbt/ZhangTT13} proposed a scalable integrated inverted index named $I^3$ which adopts the Quadtree structure to hierarchically partition the data space into cells. Besides, they presented a new storage mechanism for efficient retrieval of keyword cell and preserve additional summary information to facilitate pruning. In their other works~\cite{DBLP:conf/sigir/ZhangCT14}, they proposed an effective approach to address the top-$k$ distance-sensitive spatial keyword query by modeling it as the well-known top-$k$ aggregation problem. Moreover, a novel and efficient approach named Rank-aware CA (RCA) algorithm is designed by them to improve the effectiveness of pruning. Zheng et al.~\cite{DBLP:conf/icde/ZhengSZSXLZ15} studied interactive top-$k$ spatial keyword (IT$k$SK) query and desinged a three-phase solution focusing on both effectiveness and efficiency. In order to solve the problem of top $k$ spatial keyword search (TOPK-SK) efficiently, Zhang et al.~\cite{DBLP:journals/tkde/ZhangZZL16} proposed a novel index structure called inverted linear quadtree (IL-Quadtree) which is designed to utilize both spatial and keyword based pruning techniques to effectively reduce the search space.

João B. Rocha-Junior et al.~\cite{DBLP:conf/edbt/Rocha-JuniorN12} solved the problem of processing top-$k$ spatial keyword queries on road networks for the first time. In this type of problem, the distance between the query location and the spatial object is the shortest path, rather than Euclidean distance. They presented novel indexing structures and algorithms that are able to process such queries efficiently. Guo et al.~\cite{DBLP:journals/geoinformatica/GuoSAT15} studied continuous top-$k$ spatial keyword queries on road networks for the first time. They presented two methods that can monitor such moving queries in an incremental manner and reduce repetitive traversing of network edges for better performance.

The approaches above-mentioned is to search spatial objects with spatial information and keywords. However, they are not adequately suitable to solve the problem of top-$k$ temporal spatial image search.

\noindent\textbf{Temporal Spatial Keywords Query.} The aforementioned approaches consider only the spatial information and textual content of objects. However, temporal information is another significant dimension which should be considered in the processing of query. Temporal spatial keywords query is another important problem concerned by many researchers in recent years. Mehta et al.~\cite{DBLP:conf/gis/MehtaSSV16} proposed a novel type of spatial-temporal-keyword query named $k$CD-STK query which combines keyword search with the task of maximizing the spatio-temporal coverage and diversity of the returned top-$k$ results. Furthermore, an efficient approach which utilizes a hybrid spatial-temporal-keyword index is introduced by them to substantially improve the efficiency of query. Nepomnyachiy et al.~\cite{DBLP:conf/gir/NepomnyachiyGJM14} introduced a search framework named 3W for geo-temporal stamped documents. Their system can efficiently processes multi-dimensional queries over text, space, and time. Chen et al.~\cite{DBLP:conf/icde/ChenCCT15} consider the temporal spatial-keyword top-$k$ Subscription (TaSK) query. The TaSK query takes into account three aspects of objects: text relevance, spatial proximity and recency. They introduced a new concept Conditional Influence Region (CIR) to represent the TaSK query and proposed an algorithm for making use of the filtering conditions (of each group of queries on each spatial cell) to efficiently address the problem of TaSK. However, these approaches mentioned are just suitable to textual query, rather than image retrieval.

\noindent\textbf{Content-Based Image Retrieval.} Recently, Content-based image retrieval (CBIR for short)~\cite{DBLP:journals/pami/JingB08,DBLP:journals/tomccap/LewSDJ06,DBLP:conf/mm/WangLWZ15,DBLP:journals/tip/WangLWZ17,DBLP:conf/mm/WangLWZZ14,DBLP:conf/sigir/WangLWZZ15} is widely noted in multimedia community, which is one of the fundamental research challenges. CBIR aims to search for images through analyzing their visual contents, and thus image representation~\cite{DBLP:conf/mm/WanWHWZZL14,DBLP:journals/tip/WangLWZZH15,YANGTCYB,YANGNeurocomputing,YANGKAIS,YangPAKDD14,LYACMMM13,YANGINS,DBLP:conf/cikm/WangLZ13,DBLP:journals/tnn/WangZWLZ17,DBLP:journals/pr/WuWLG18}. Local feature representations such as the bag-of-visual-words (BoVW) models~\cite{DBLP:conf/iccv/SivicREZF05,NNLS2018,DBLP:journals/ivc/WuW17} applying local feature descriptors such as SIFT~\cite{DBLP:conf/iccv/Lowe99,DBLP:journals/ijcv/Lowe04,DBLP:journals/cviu/WuWGHL18} and SURF~\cite{DBLP:conf/eccv/BayTG06,TC2018}. BoVWs represents an image by a vector of visual words which is constructed by vector quantization of feature descriptors~\cite{DBLP:conf/iccv/SivicZ03}. Many researches worked for this issue over the years. For example, Irtaza et al.~\cite{DBLP:journals/mta/IrtazaJAC14} proposed a neural network based architecture for content based image retrieval. In order to improve the capabilities of their approach, they designed an efficient feature extraction algorithm based on the concept of in-depth texture analysis. Bunte et al.~\cite{DBLP:journals/pr/BunteBJP11} used two different methods to learn favorable feature representations Limited Rank Matrix Learning Vector Quantization (LiRaMLVQ) and Large Margin Nearest Neighbor (LMNN). Zhao et al.~\cite{DBLP:conf/mm/ZhaoYYZ14} studied affective image retrieval and the performance of different features on different kinds of images. Xie et al.~\cite{DBLP:conf/icmcs/XieYH14} presented a hypergraph-based framework integrating image content, user-generated tags and geo-location information into image ranking problem. Zhu et al. studied the problem of content-based landmark image search and proposed multimodal hypergraph (MMHG) to characterize the complex associations between landmark images. Based on it, they designed a novel content-based visual landmark search system to facilitate effective image search. However, these methods do not consider the temporal information and geographical proximity of images. Thus they are not suitable to the problem of temporal spatial image search.
\section{System Overview}
\label{overview}

This section first gives an overview of architecture of spatial temporal image search system, then introduces the problem definition of spatial temporal image search.
\subsection{System Architecture}
The proposed spatial temporal image search system consists of three components, namely, preprocess, update, and query modules, as what is showed in Figure ~\ref{fig:fig3architecture}.

\begin{figure*}
\newskip\subfigtoppskip \subfigtopskip = -0.1cm
\begin{minipage}[b]{0.99\linewidth}
\begin{center}
     \subfigure[]{
     \includegraphics[width=0.90\linewidth]{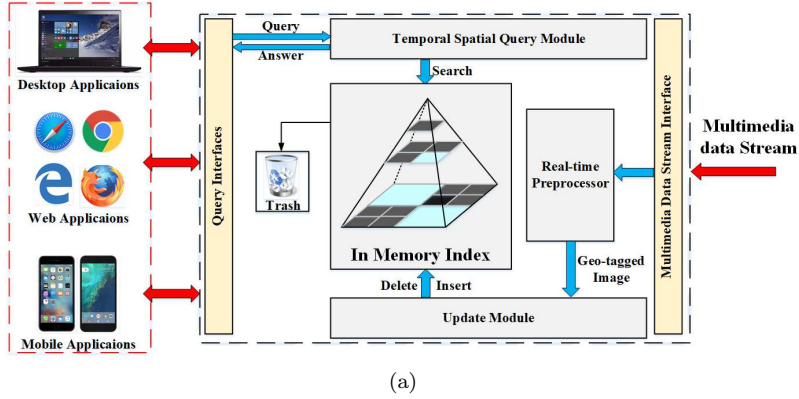}
     }
   \captionsetup{justification=centering}
       \vspace{-0.2cm}
\caption{Spatial temporal image system architecture}
\label{fig:fig3architecture}
\end{center}
\end{minipage}
\label{fig:k}
\end{figure*}

\textbf{Preprocess Module.}
This module used to receive the incoming spatial temporal image, extracts the location of each geo-temporal image, and forwards each geo-temporal image along with its extracted location
to the update module with the form: (id, location, timestamp,
visual word list)which describes the geo-temporal image's identifier, geo-location, issuing
time, and image contents. Location is either a precise latitude
and longitude coordinates or the center of a Minimum Bounding
Rectangle. This module will not be further discussed in the following sections, because we directly use existing preprocessing programme to extract the location information from public datasets.

\textbf{Update Module.} The update module ensures all incoming
spatial temporal image can be inserted to in-memory indexes as soon as possible, and all incoming
spatial temporal image queries can be answered accurately from the in-memory indexes with minimum possible memory consumption. This is
done through two main tasks: (1) Inserting
newly coming spatial temporal image into the latest in-memory index structure; (2) Deleting expire spatial temporal image from the most forward in-memory index structure without sacrificing the query answer quality

\textbf{Query Module.}
Given a spatial temporal image search query, the query module
employs spatio-temporal vusial pruning techniques that reduce
the number of visited images to return the final answer. As the query module just retrieves the images in the index, the accuracy of result is mainly decided by the decisions taken at the update
module on which spatial temporal image will expire from the in-memory
index.

\subsection{Problem Statement}

In this section, we present problem definition and necessary preliminaries of top $k$ spatial temporal image search. Table~\ref{tab:skt notation} below summarizes the mathematical notations used throughout this section.

\begin{table}
	\centering
    \small
	\begin{tabular}{|p{0.27\columnwidth}| p{0.62\columnwidth} |}
		\hline
		\textbf{Notation} & \textbf{Definition} \\ \hline\hline
		~$I(q)$             & s geo-temporal image (query)                                \\ \hline
        ~$I.\psi(q.\psi)$   & a set of visual words used to describe I (query q)                                \\ \hline
	  	~$I.loc(q.loc)$     & location of the image I (query q)                                 \\ \hline
        ~$I.t(q.t)$         & timestamp of the image I (query q)  \\ \hline
		~$\mathcal{V}$      & vocabulary                              \\ \hline
		~$v$                & a visual word in $\mathcal{V}$                                \\ \hline	
    	~$l$  & the number of query visual word in $q.\psi$ \\\hline
        ~$k$   & the number of results should be returned  \\\hline
        ~$\omega_{1}$  & the preference parameter to balance the spatial proximity, visual relevance and temporal recency \\ \hline
        ~$\omega_{2}$  & the preference parameter to balance the spatial proximity, visual relevance and temporal recency \\ \hline
        ~$\omega_{3}$  & the preference parameter to balance the spatial proximity, visual relevance and temporal recency \\ \hline
        ~$f_s(I,q)$     & the spatial proximity between $I.loc$ and $q.loc$                                \\ \hline
        ~$f_v(I,q)$         & the temporal recency between $I.\psi$ and $q.\psi$ \\ \hline
        ~$f_t(I,q)$         & the visual relevance between $I.t$ and $q.t$ \\ \hline
        ~$f_{stv}(I,q)$         & the spatial temporal visual ranking score between $I$ and $q$\\ \hline
	\end{tabular}
    \vspace{2mm}
    \caption{Notations} \label{tab:skt notation}	
    \vspace{-6mm}
\end{table}

In this section, $\mathcal{O}$ denotes a sequence of incoming stream geo-temporal images. A \textbf{geo-temporal image} is an image message with geo-location and timestamp, such as geo-temporal photos in Flickr. Formally, a geo-temporal image is modeled as $I$ = $<\psi, loc, t_c>$, where $o.\psi$ denotes a set of distinct visual words from a visual vocabulary set $\mathcal{V}$, $I.loc$ represents a geo-location with latitude and longitude, and $I.t$ represented the creation timestamp of the image.

\begin{definition}[\textbf{Top-$k$ Spatial Temporal Image Query}] \label{def:skt sktq}
A top-$k$ spatial temporal image query $q$ is defined as q = $<\psi,$ $loc,$ $t,$ $k>$, where $q.\psi$ is a set of distinct visual words extracted from query image, $q.loc$ is the query location, $q.t$ is the user submitted timestamp, $k$ is the number of the result user expected.
\end{definition}

\begin{definition} [\textbf{Spatial Proximity $f_s(q,I)$}]
Let $\delta_{max}$ denote the maximal distance in the space, the spatial relevance between the query $q$ and the image $I$, denoted by $f_s(q, I)$, is defined as $\frac{\delta(q.loc, I.loc)}{\delta_{max}}$.
\end{definition}

Similar to ~\cite{DBLP:journals/pvldb/CongJW09}, we adopt the language model based function
to measure the visual words' relevance of the image regarding the
query $q$, which is defined as follows.

\begin{definition} [\textbf{Visual Relevance $f_v(q,I)$}]
Let $w_{t,I}$ denote the weight of the visual word $v$ regarding image $I$, and
\begin{eqnarray}
\label{eqn:ilq rank1}
w_{v,I} &=& (1 - \xi)\frac{tf_{v,I}}{|I.\psi|} + \xi \frac{tf_{v,\mathcal{I}}}{|\mathcal{I}|},
\end{eqnarray}
where $tf_{v,I}$ and $tf_{v,\mathcal{I}}$ are the \texttt{term frequency} of $v$ in $I.\psi$ and
$\mathcal{I}$ respectively. Here, $\mathcal{I}$ represents the visual word information of all images in the whole dataset and $\xi$ is a smoothing parameter.
Then the visual relevance between $q$ and $I$ is defined as follows.
\begin{eqnarray}
\label{eqn:ilq rank2}
f_t(q,I) &=& 1 - \frac{ \prod_{v \in q.\psi } w_{v,I}}{\gamma_{max}}
\end{eqnarray}
where $\gamma_{max}$ is used for normalization.
\end{definition}

\begin{definition} [\textbf{Temporal Recency $f_t(q,I)$}]
The temporal recency between the query $q$
and the image $i$, denoted by $f_t(q,I)$, is calculated by the following
exponential decay function:
\begin{equation}\label{eq:skt function}
f_t(q,I)=D^{-(q.t-I.t_c)}.
\end{equation}
where $D$ is base number that determines the rate of the recency
decay. The function is monotonically decreasing with $q.t-I.t_c$.
It is introduced in ~\cite{DBLP:conf/sigmod/LuLC11} and is applied (e.g., ~\cite{DBLP:conf/cikm/AmatiAG12}) as
the measurement of recency for stream data. Based on the
experimental studies ~\cite{DBLP:conf/sigir/EfronG11}, the exponential decay function has
been shown to be effective in blending the recency and text
relevancy of objects. Thus, we use the exponential decay function to blend the recency and visual
relevancy of image.
\end{definition}

Based on the spatial proximity, visual relevance and temporal recency
between the query and the spatial temporal image, the \textbf{Spatial-temporal Ranking Score} of an image $I$ regarding the query $q$ can
be defined as follows.

\begin{definition} [\textbf{Spatial Temporal Visual Score} $f(q,I)$]
Let $\omega_{1}$, $\omega_{2}$, $\omega_{3}$  be the preference parameter specifying the trade-off among the spatial proximity, visual relevance and temporal recency, $\omega_{1}+\omega_{2}+\omega_{3}=1$ and  $\omega_{1}, \omega_{2}, \omega_{3}> 0$ we have
\begin{equation}\label{eq:skt function}
f(q,I)=\omega_{1}*f_s(q,I) + \omega_{2}*f_v(q,I) + \omega_{3}*f_t(q,I).
\end{equation}
Note that the images with the \textbf{small score values} are preferred (i,e., ranked higher).
\end{definition}

\begin{definition}[\textbf{Spatial Temporal Image Search}] \label{def:skt skts}
Given a set of geo-textual Image $I$ and a spatial temporal image query $q$, we aim to find the top $k$ geo-temporal images with \textbf{smallest} spatial temporal visual Score score.
\end{definition}

In addition, we require that only images $I$ with
$f_v(q,I) > 0$ can be returned as query results, so as to avoid
giving completely irrelevant answers to users’ queries. This
implies that $I$ and $q$ should have at least one common term.

In the section hereafter, we abbreviate the geo-temporal image and the geo-temporal query as $image$ and $query$ respectively, if there is no ambiguity. We assume there is a total order for visual words in $\mathcal{V}$, and the words in each query and image are sorted accordingly.
\section{Baseline}
\label{baseline}

Before proceeding to present the proposed solution, we discuss
the possibility of using conventional techniques for the processing of spatial temporal image queries. In the following, we develop two baselines by utilizing existing techniques: Inverted File Append (\ifa for short) and Spatial Temporal Visual Inverted Index(\stvii).

\subsection{Inverted File Append}
A nice property of the inverted indexing structure is that, for a given query $q$, only the objects containing at least one
query keyword will be involved in the search. Thus, inverted file is widely  used to process textual queries efficiently[14].

To adopt inverted file for spatial temporal image search, the
simplest approach is to treat each geo-temporal image as a document, and sort the entries in
each posting list in ascending order of the corresponding
geo-temporal images’ timestamps. This approach is highly efficient in
terms of geo-temporal image insertions, as new geo-temporal image can be easily
appended to the ends of posting lists without affecting the
ordering of entries. In terms of query efficiency, however, the
aforementioned approach faces a great challenge. This is because our
spatial temporal visual ranking score $f(q,I)$ evaluates a geo-temporal image $I$ based on three factors: its spatial proximity $f_s(q,I)$, visual relevance $f_v(q,I)$, and temporal recency $f_t(q,I)$. If the entries in a
posting list are sorted in ascending order of timestamps, the
corresponding geo-temporal image would be in ascending order of
temporal recency, regardless of their spatial proximity or visual relevance.
In other
words, the entry order does not provide any hint on the overall
score of each geo-temporal image. As a consequence, when answering
a query $q$�, we have to examine most of entries in all posting lists
relevant to $q$�, since the omission of any entry may render
the query results incomplete.

\subsection{Spatial Temporal Visual Inverted Index}
To the best of our knowledge, there is no index in the literature
that can filter objects taking into account the three criteria: spatial,
temporal and visual. Thus, in this section, we present an hybrid index, namely spatial temporal visual inverted index(\stvii), which may filter geo-temporal images taking into account the spatial, temporal and visual information simultaneously.

The \stvii is similar to traditional spatial keyword search index \sii. The major difference is how to use the spatial index. \sii uses a R-tree to select objects that are
spatially relevant, but \stvii employs a \tdr~\cite{DBLP:conf/icmcs/TheodoridisVS96}, which takes spatial and temporal dimensions into consideration synchronously.

Obviously, the \tdr can filter geo-temporal that are spatially and temporally
unrelated in the early phase of the query processing. Unfortunately, it may meets poor performance in some situation. The minimum bounding regions of \tdr are 3D rectangles, because the geo-temporal images stored in \tdr has three dimensions: time, latitude, and longitude. Meanwhile, space and time are
not correlated dimensions. Thus, the minimum bounding regions of \tdr may cover large areas
of the space, which may result in large area overlap. In this scenario, long periods of time or large spatial regions may give rise to poor performance.

\section{Spatial Temporal Visual Indexing}
\label{StvIndex}
Due to massive amount of spatial temporal images and queries are being generated at an unprecedented scale, three main objectives have to be satisfy in our spatial temporal visual indexing. First, the proposed index has to be able to handle
high arrival rates of incoming spatial temporal images. Second,expired objects can be deleted from its index with the approximative rate
as insertion. Third, a large number of unpromising spatial temporal images can be filtered at a cheap cost.

\subsection{Our proposed: Hierarchical Information Quadtree}
\label{sec:skt framwork}

Based on the above requirement, in this sub-section, we present a \textbf{H}ierarchical \textbf{I}nformation \textbf{Q}uadtree (\hiq for short)
that supports update at high arrive rate and provides the following required functions for spatial temporal image search and ranking: \MyRoman{1})\textbf{temporal filtering}: all the temporally irrelevant trees, nodes and images have to be accessed as late as possible to follow the chronological order; \MyRoman{2})\textbf{spatial filtering}: all the spatially irrelevant nodes have to be filtered out as early as possible to shrink the search space; \MyRoman{3})\textbf{visual word filtering}: all the visually irrelevant trees, nodes and images have to be discarded as early as possible to cut down the search cost;  and  \MyRoman{4})\textbf{relevance computation and ranking}: since only the top-$k$ images are returned and $k$ is expected to be much smaller than the total number of similar images, it is desirable to have an incremental search process that integrates the computation of the joint relevance, and image ranking seamlessly so that the search process can stop as soon as the top-$k$ images are identified.
Figure ~\ref{fig:fig4hiq} shows two levels of \hiq, namely, temporal segment and inverted quadtree.

\begin{figure*}
\newskip\subfigtoppskip \subfigtopskip = -0.1cm
\begin{minipage}[b]{0.99\linewidth}
\begin{center}
     \subfigure[]{
     \includegraphics[width=0.90\linewidth]{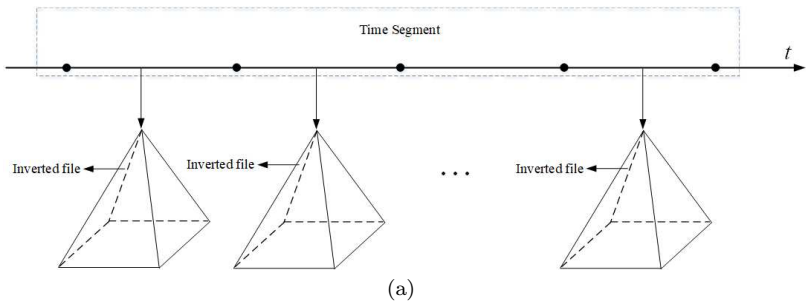}
     }
   \captionsetup{justification=centering}
       \vspace{-0.2cm}
\caption{Example of hierarchical information quadtree}
\label{fig:fig4hiq}
\end{center}
\end{minipage}
\label{fig:k}
\end{figure*}

\textbf{Temporal Segment}

To achieve fast insertion and deletion, all spatial temporal images are temporally partitioned into successive disjoint index segments. For example, each segment only indexes
the data of $T$ hours. Thus, it makes sure that the newly
incoming spatial temporal images are always inserted in the most recent segment. Once the segment spans $T$ hours of data, the segment is terminated and a new empty segment
is generated to insert the new spatial temporal image.

\textbf{Inverted Quadtree}

To support high arrival rates of incoming objects, space-partitioning index (e.g., \qdtree~\cite{DBLP:journals/cacm/Gargantini82,DBLP:conf/pakdd/WangLZW14,DBLP:journals/corr/abs-1708-02288,DBLP:conf/icde/ZhangZZL13,DBLP:journals/tkde/ZhangZZL16,TSMCS2018}, and \pyr~\cite{DBLP:conf/pods/ArefS90,DBLP:conf/ijcai/WangZWLFP16,DBLP:journals/pr/WuWGL18,TII2018}) is more famous than object-partitioning index (e.g., \rtree).
As space-partitioning index is more suitable to high update system because of its disjoint space decomposition policy, while the shape of object-partitioning index is highly affected by the rate and order of incoming data, which may trigger a large number of node splitting and merging. Meanwhile, inverted file, which is the most efficient index for text information retrieval, can easily extend to visual words. Thus, we propose a hybrid indexing structure, namely inverted quadtree,
that utilizes both indexing structures in a combined fashion.

The inverted quadtree is essentially an quadtree, each node of which is enriched
with reference to an inverted file for the images contained in
the sub-tree rooted at the node. In particular, each node of an inverted quadtree contains all spatial, temporal, and visual words information; the first is in the form of a rectangle, the second is in the form of timestamp, and the last is in the form of an inverted file.

More formally, the leaf node of \iqdtree has the form ($I$, $r$, $t$). where $I$ refers to a set of images belonged to current node, $r$ is the area covered by current node, and $t$ is the latest timestamp aggregated from the images. A leaf node also contains a pointer to a visual inverted file for the visual words of the images being indexed.

An visual inverted file consists of a vocabulary for all distinct visual words in a collection of images and a set of posting lists related to this vocabulary. Each
posting list is a sequence of visual pairs $<ip, w_{I,v}>$, where $ip$ refers to
a image containing visual words $v$, and $w_{I,v}$ is the weight of term
$v$ in image $I$.

An inner node contains a number of entries of the form ($cp$, $r$, $t$). $cp$ are the address of the children nodes, $r$ is the area covered by current node, and $t$ is the latest timestamp aggregated from its children nodes. A inner node also contains a pointer to a visual inverted file which is aggregated from the visual inverted files from its child node. This inverted file includes all images in the entries of current node, enabling us to estimate a bound of the visual relevancy to a
query of all images contained in the subtree rooted at current node. The
weight of each visual word $v$ in the inverted index
is the maximum weight of the visual word in the images contained in
the subtree rooted at current node. Fig depicts an \iqdtree indexing structure.


\subsection{Processing of Spatial Temporal Image Queries}\label{sec:skt query}

\begin{algorithm}
\begin{algorithmic}[1]
\footnotesize
\caption{\bf Spatial Temporal Image Search($q$, $k$, $\mathcal{I}$) }
\label{alg:skt search}
\INPUT $q~:$ the spatial temporal image query, $k~:$ the number of image return,
   $\mathcal{I}~:$ current \hiq index
\OUTPUT $\mathcal{R}:$ top-$k$ query result results

\STATE $\mathcal{R}:= \emptyset$; $\mathcal{H} = \emptyset$, $\lambda_{max} = \infty$

\STATE $\mathcal{H}\leftarrow$ new a min first heap

\STATE build frequency signature for query
\label{alg:skt_build fsig}
\STATE $\mathcal{H}$.Enqueue($\mathcal{I}.root, MIND_{st}(q,\mathcal{I}.root)$)
\label{alg:skt push root}

\WHILE{ $\mathcal{H} \not = \emptyset$ }
     \label{alg:skt_loop_s}
     \STATE $e \leftarrow$ the node popped from $\mathcal{H}$
     \IF{$e$ is a leaf node}
        \label{alg:skt_leaf s}

            \FOR{ each image $I$ in node $e$}
                \IF{$f_{stv}(q, I) \leq \lambda_{max}$}
                            \STATE $\lambda_{max} \leftarrow f_{stv}(q, I)$
                            \STATE update $\mathcal{R}$ by $(I, f_{stv}(q, I))$
                             \label{alg:skt_leaf e}
                \ENDIF
            \ENDFOR

     \ELSE
        \FOR{ each child $e'$ in node $e$}
            \label{alg:skt_non-leaf s}
            \IF{$MIND_{stv}(q, e') \leq \lambda_{max}$}
               \STATE $\mathcal{H}$.Enqueue($e',  MIND_{stv}(q, e')$)
               \label{alg:skt_non-leaf e}
            \ENDIF
        \ENDFOR
     \ENDIF

    \STATE process the root node of next \hiq
    \label{alg:skt_process next}
\ENDWHILE
\RETURN{$\mathcal{R}$}
\label{alg:skt_return}
\end{algorithmic}
\end{algorithm}

We proceed to present an important metric, the minimum spatial temporal visual score $MIND_{stv}$, which will be used in the query processing. Given a query $q$ and a node $N$ in the \hiq, the metric $MIND_{stv}$ offers a lower bound on the actual spatial temporal visual score between query $q$ and the images enclosed in the rectangle of node $N$. This bound can be used to order and efficiently prune the paths of the search space in the \hiq.

\begin{definition} [$MIND_{stv}(q,N)$]
The score of a query point $q$ from a node $N$ in
the \hiq, denoted as $MIND_{stv}(q,N)$, is defined as follows:
\begin{equation}\label{eq:skt mindst}
\begin{aligned}
MIND_{stv}(q,N)=\omega_{1}*MIND_{s}(q, N) + \\
 \omega_{2}*MIND_{v}(q,N) + \omega_{3}*MIND_{t}(q,N)
\end{aligned}
\end{equation}
where $\omega_{1}$, $\omega_{2}$, and $\omega_{3}$are the same as in Equation \ref{eq:skt function};
$MIND_{s}(q, N)$ is the minimum Euclidian distance between $q.loc$ and $N.r$, $MIND_{v}(q, N)$ is the minimum visual relevance between $q.\psi$ and $N.\psi$,
$MIND_{t}(q, N)$ is the minimum time recency between $q.t$ and $N.t$.
\end{definition}

A salient feature of the proposed \hiq structure is that
it inherits the nice properties of the \qdtree for query processing.

\begin{theorem}\label{lemma:skt point_prune}
Given a query point $q$, a node $N$, and a set of geo-temporal images $\mathcal{I}$ in node $N$, for any $I\in\mathcal{I}$, we have $f_{stv}(q,N)\leq DIST_{stv}(q,I)$.
\end{theorem}
\begin{proof}
Since geo-temporal image $I$ is enclosed in the rectangle of node $N$, the
minimum Euclidian distance between $q.loc$ and $N.r$ is
no larger than the Euclidian distance between $q.loc$ and $I.loc$:

\begin{equation*}
MISD_{s}(q.loc, N.r)\leq f_{s}(q.loc, I.loc)
\end{equation*}
For each timestamp $t$, $N.t$ is the maximum value $\mathcal{I}.t$ of all the geo-temporal images in node $N$. Hence:
\begin{equation*}
MISD_{t}(q.loc, N.r)\leq f_{t}(q.loc, o.loc)
\end{equation*}
For each visual word $v$, $w_{N.v}$ (the weight of the visual word in $N$, which is the inverted file of node $N$) is the maximum value $w_{I.v}$ of all the geo-temporal images in node $N$. Thus:
\begin{equation*}
MISD_{v}(q.\psi, N.\psi)\leq f_{t}(q.\psi, I.\psi)
\end{equation*}
According to Equation \ref{eq:skt function} and Equation \ref{eq:skt mindst}, we obtain:
\begin{equation*}
MIND_{stv}(q,N)\leq f_{st}(q,I)
\end{equation*}
thus completing the proof.
\end{proof}

When searching the \hiq for the $k$ objects nearest to a
query $q$, one must decide at each visited node of the \hiq
which entry to search first. Metric $MIND_{ST}$ offers an approximation
of the spatial-temporal ranking score to every entry in the node and, therefore, can
be used to direct the search. Note that only node satisfied the constraint of query keywords need to be loaded into memory
and compute $MIND_{ST}$.

To process spatial temporal image queries with \hiq framework, we
exploit the best-first traversal algorithm for retrieving
the top-k objects. With the best-first traversal algorithm, a priority
queue is used to keep track of the nodes and objects that have yet to
be visited. The values of $f_{st}$ and $MIND_{st}$ are used as the keys
of objects and nodes, respectively.

When deciding which node to visit next, the algorithm picks the
node $N$ with the smallest $MIND_{st}(q, N)$ value in the set of all
nodes that have yet to be visited. The algorithm terminates when $k$
nearest objects (ranked according to Equation \ref{eq:skt function}) have been found.

Algorithm~\ref{alg:skt search} illustrates the details of the \hiq based spatial temporal image query. A minimum heap $\mathcal{H}$ is employed to keep the \hiq's nodes where the key of a node is its minimal spatial temporal visual ranking score. For the input query, we calculate its frequency signature in Line~\ref{alg:skt_build fsig}. In Line ~\ref{alg:skt push root}, we find out the root node of current time segment, calculate the minimal spatial temporal visual ranking score for the root node, and then pushed the root node into the $\mathcal{H}$. The the algorithm executes the while loop (Line~\ref{alg:skt_loop_s}-\ref{alg:skt_process next})until the top-$k$ results are ultimately reported in Line~\ref{alg:skt_return}.


In each iteration, the top entry $e$ with minimum spatial temporal visual ranking score is popped from $\mathcal{H}$. When the popped node $e$ is a leaf node(Line~\ref{alg:skt_leaf s}), for each signature in node $e$, we will iterator  extract the image and check whether its spatial temporal visual score is less than $\lambda_{max}$. If its score is not larger than $\lambda_{max}$, we push $I$ into result set and add update $\lambda_{max}$. When the popped node $e$ is a non-leaf node(Line~\ref{alg:skt_non-leaf s}), a child node $e^{'}$ of $e$ will be pushed to $\mathcal{H}$ if its minimal spatial temporal visual ranking score between $e^{'}$ and $q$, denoted by $MIND_{stv}(q, e^{'}$, is not larger than $\lambda_{max}$ (Line~\ref{alg:skt_non-leaf s}-~\ref{alg:skt_non-leaf e}). We process the root node of next interval in Line~\ref{alg:skt_process next}. The algorithm terminates when $\mathcal{H}$ is empty and the results are kept in $\mathcal{R}$.

\section{PERFORMANCE EVALUATION}
\label{perform}

In this section, we present the results of a comprehensive performance study to evaluate the effectiveness and efficiency of our techniques proposed in this paper.

\noindent\textbf{Workload.} A workload for this experiment consists of 100 input queries, and the average query response time are employed to evaluate the performance of the algorithms. The query locations are randomly selected from the underlying dataset. By default, the number of query visual keywords varies from 10 to 200, the number of results $k$ grows from 10 to 100, and the preference parameter $\omega$ changes from 0:1 to 0:9.

Experiments are run on a PC with Intel i7 7700K 4.20GHz CPU and 16GB RAM running Ubuntu 16.04 LTS Operation System. All algorithms in the experiments are implemented in Java. For a fair comparison, we tune the important parameters of the competitor algorithms for their best performance. Particularly, the node capacity of all algorithms is set to 100. Our measures of performance include insertion time, deletion time, storage overhead, number of node access and response time. The rest of this section evaluates index maintenance and query processing.

\begin{table}
\centering
\caption{Information of datasets} \label{tab:prob}

\centering
\begin{tabular}{cccc}

\hline
Datasets& Number of Images& Dist. Visual Words Number& Avg. Visual Words Number\\
\hline
200K&	200000&   	616347&  	124.7\\
400K&	400000&	    613940&  	118.2\\
600K&	600000&    	613026&  	114.3\\
800K&	800000&	    607401& 	128.6\\
1M&	    1000000&	612905& 	134.8\\
\hline
\end{tabular}

\end{table}

\noindent\textbf{Dataset.} We first evaluate the scalability and performance of our system on an image dataset of over one million images crawled from the photo-sharing site, Flickr, using Oxford landmarks as queries. For the scalability and performance evaluation, we randomly sampled five sub datasets whose sizes vary from 100,000 to 2000,000 from the image dataset.

\subsection{Index Maintenance}
In this subsection, we evaluate the insertion time, deletion time, storage overhead of all the algorithms.

\noindent\textbf{Evaluation on insertion time.} Fig. gives the performance when varying the arrival rate from 200 to 3200. It is clearly that the average insertion time of these three algorithm gradually grows with the increasing of arrival rate. Both HIQ and IFA have a better performance than STVI, shown as Fig.(a). As IFA adopts a simple data structure, the insertion need less time than HIQ. Fig.(b) illustrates that the average insertion time of HIQ and STVI with varying node capacity from 100 to 500. As HIQ has no node capacity, we just compare HIQ and STVI. Apparently, the time cost of STVI is nearly 4 times of HIQ.

\noindent\textbf{Evaluation on deletion time.} Fig indicates that the average deletion time of these three algorithm. With the rising of arrival rate from 200 to 3200, all of them increase step by step. Like the situation of evaluation on insertion time, the performance of IFA is the best due to its simple structure. the deletion time of our method is less than STVI. We can find out from the Fig.(b) that when varying the node capacity from 100 to 500, the performance of HIQ and STVI wave slightly. The former fluctuate between 70 and 75, and the latter  As expected, HIQ has the best performance when node capacity changes.
%

\subsection{Time Evaluation}
In this subsection, we evaluate the query response time of all the algorithms.

\noindent\textbf{Effect of the number of query visual words.} To investigate the response time of IFA, STVI and our HIQ algorithm under queries with different number of visual words, we increase the number of visual query visual words $l$ from 10 to 200. Fig.(a) evaluates the response time of these three methods where the number of query visual words varies from 10 to 200. Not surprisingly, all of them increase gradually with the rising of $l$ and the performance of HIQ is the best. The time cost of HIQ and STVI grow faster than IFA when $l<100$.

\noindent\textbf{Effect of the number of returned results.} We increase the number of results $k$ from 10 to 200 and evaluate the time cost of these three methods. In Fig.(b), our methods HIQ shows superior performance in comparison with other algorithms, which increases step by step with the growth of $k$. Clearly, The trend of STVI is similar to HIQ. On the other hand, the response time of IFA is almost unchanged.

\noindent\textbf{Effect of dataset size.} We study the response time under different sizes of image datasets $n$. The experimental results are shown in Fig.(c). It is obvious that with the increasing of $n$, the time cost of IFA, STVI and HIQ ascend by degrees. When $n>600k$, the climbing of them is slow down. Not surprisingly, the performance of our method is the best of them.


\noindent\textbf{Effect of weight $\omega$.} In the next experiment, we vary $\omega_1$ from 1/7 to 5/7, setting $\omega_2$ = $\omega_3$ = $(1 - \omega_1)/2$. Fig.(d) demonstrates the processing cost of each method as a function of $\omega_1$. We can observe that the performance of these three methods are practically unchanged in the interval $1/7<\omega_1<5/7$. Like the evaluation above-mentioned, the performance of HIQ is the highest among them. Fig.(e) illustrate the results of evaluation on varying $\omega_2$. We can see that with the growth of $\omega_2$ all of the performance of these three algorithms slightly slow down. In Fig.(f), the response time of HIQ and IFA gently decrease and the performance of STVII is almost unchanged.

\section{Conclusion}
\label{con}
To the best of our knowledge, this is the first work to study the problem of spatial temporal image queries over streaming spatial temporal multimedia stream, which has a wide spectrum of application. To tackle with this problem, we propose a novel spatial temporal visual indexing structure, namely \hiq, efficiently organize a massive number of streaming spatial temporal images such that each incoming query submitted by users can rapidly find out the top-$k$ results. An efficient spatial temporal image search algorithm based on \hiq is designed to deal with this problem. Extensive experiments
demonstrate that our technique achieves a high throughput performance over streaming spatial temporal multimedia data.

\textbf{Acknowledgments:} This work was supported in part by the National Natural Science Foundation of China
(61379110, 61472450, 61702560), the Key Research Program of Hunan Province(2016JC2018), and project 2018JJ3691 of Science and Technology Plan of Hunan Province.



\bibliographystyle{spmpsci}      

\bibliography{ref}

\end{document}